\documentclass[sigconf,nonacm,anonymous=false]{acmart}

\newcommand\vldbdoi{XX.XX/XXX.XX}

\newcommand\vldbavailabilityurl{}
\newcommand{\tinyskip}{\vspace*{-.5\baselineskip}}
\newcommand{\miniskip}{\vspace*{-.5\baselineskip}}
\newcommand{\shrink}{\vspace*{-0.7\baselineskip}}

\newtheorem{lemma}{Lemma}

\newtheoremstyle{mydef}% style name
{2ex}% above space
{2ex}% below space
{\itshape}% body font
{}% indent amount
{\scshape}% head font
{: }% post head punctuation
{0.5em}% space after theorem head
{}% head spec 
\theoremstyle{mydef}

\newtheoremstyle{remark}% style name
{1ex}% above space
{1ex}% below space
{\normalfont}% body font
{}% indent amount
{\scshape}% head font
{: }% post head punctuation
{0.5em}% space after theorem head
{}% head spec 
\theoremstyle{remark}

\newtheoremstyle{assumption}% style name
{2ex}% above space
{2ex}% below space
{\normalfont}% body font
{}% indent amount
{\scshape}% head font
{: }% post head punctuation
{0.5em}% space after theorem head
{}% head spec 
\theoremstyle{assumption}

\usepackage{algorithm2e}

\usepackage{multirow}
\usepackage{booktabs} % midrule and toprule
\usepackage{rotating}%sideways

\usepackage{subcaption}

\usepackage[normalem]{ulem} % sout

\begin{document}
	
	\title{Efficiently Transforming Tables for Joinability}

	\author{Arash Dargahi Nobari}
	\affiliation{%
		\institution{University of Alberta}
		%\streetaddress{P.O. Box 1212}
		\city{Edmonton}
		\state{Canada}
	}
	\email{dargahi@ualberta.ca}
	
	\author{Davood Rafiei}
	%\orcid{0000-0002-1825-0097}
	\affiliation{%
		\institution{University of Alberta}
		%\streetaddress{1 Th{\o}rv{\"a}ld Circle}
		\city{Edmonton}
		\country{Canada}
	}
	\email{drafiei@ualberta.ca}

	\begin{abstract}
		Data from different sources rarely conform to a single formatting even if they describe the same set of entities, and this raises concerns when data from multiple sources must be joined or cross-referenced.
		Such a formatting mismatch is unavoidable when data is gathered from various public and third-party sources. Commercial database systems are not able to perform the join when there exist differences in data representation or formatting, and manual reformatting is both time consuming and error-prone.
		We study the problem of efficiently joining textual data under the condition that the join columns are not formatted the same and cannot be equi-joined, but they become joinable under some transformations. 
		The problem is challenging simply because the number of possible transformations explodes with both the length of the input and the number of rows, even if each transformation is formed using very few basic units. We show that an efficient algorithm can be developed based on the common characteristics of the joined columns, and develop one such algorithm over a rich set of basic operations that can be composed to form transformations. We compare both the coverage and the running time of our algorithm to a state-of-the-art approach, and show that our algorithm covers every transformation that is covered in the state-of-the-art approach but is a few orders of magnitude faster, as evaluated on various real and synthetic data.
	\end{abstract}
	
	% 	\keywords{\todo{@TODO: Do we need keywords?}}
	
	\maketitle
	
%	%%% VLDB block start %%%
%	\pagestyle{\vldbpagestyle}
%	\begingroup\small\noindent\raggedright\textbf{PVLDB Reference Format:}\\
%	\vldbauthors. \vldbtitle. PVLDB, \vldbvolume(\vldbissue): \vldbpages, \vldbyear.\\
%	\href{https://doi.org/\vldbdoi}{doi:\vldbdoi}
%	\endgroup
%	\begingroup
%	\renewcommand\thefootnote{}\footnote{\noindent
%		This work is licensed under the Creative Commons BY-NC-ND 4.0 International License. Visit \url{https://creativecommons.org/licenses/by-nc-nd/4.0/} to view a copy of this license. For any use beyond those covered by this license, obtain permission by emailing \href{mailto:info@vldb.org}{info@vldb.org}. Copyright is held by the owner/author(s). Publication rights licensed to the VLDB Endowment. \\
%		\raggedright Proceedings of the VLDB Endowment, Vol. \vldbvolume, No. \vldbissue\ %
%		ISSN 2150-8097. \\
%		\href{https://doi.org/\vldbdoi}{doi:\vldbdoi} \\
%	}\addtocounter{footnote}{-1}\endgroup
%	%%% VLDB block end %%%
%	
%	%%% do not modify the following VLDB block %%
%	%%% VLDB block start %%%
%	\ifdefempty{\vldbavailabilityurl}{}{
%		\vspace{.3cm}
%		\begingroup\small\noindent\raggedright\textbf{PVLDB Artifact Availability:}\\
%		The source code, data, and/or other artifacts have been made available at \url{\vldbavailabilityurl}.
%		\endgroup
%	}
%	%%% VLDB block end %%%
	
	\section{Introduction}
With the growing extent of data available in the public domain and from third-party sources, many organizations find data outside enterprise databases relevant in their decision making process and want to integrate such data sources with their internal databases.
However, integrating data from different sources can be challenging due to formatting mismatches and sometimes the lack of schema.
When data are obtained from different sources, there is no guarantee that the same pieces of data (e.g., name, phone, address) from two different sources follow the same formatting.
Even different units within the same organization, when operating independently, may format the same data differently. %For example, one department may format the name as ``last name, first name'' whereas another department may format it as ``first name last name.''
For example, a phone number can be formatted as \texttt{\small\mbox{(780) 432-3636}} by one source, \texttt{\small\mbox{+1 780 432-3636}} by another source and \texttt{\small\mbox{1-780-432-3636}} by a third source.
%Most commercial database systems only support equality join, which is applicable for well-curated data sources, which may be governed under the same authority and follow the same formatting. 
Manually performing \textit{extract, transform, load} (ETL) for each source can be tedious and time consuming. An automated join operation is considered an important step towards the vision of \textit{self-service data preparation}, which is estimated to be over \$1 billion software market~\cite{Gartner-2016}.

\begin{figure*}[htb]
	\resizebox{1\linewidth}{!}{
		\begin{tabular}{lll}
			\hspace{-0.2cm}
			\begin{tabular}{l|l}
				{\bf Course} & {\bf Contact email}\\ \hline
				CMPUT 291 & drafiei@ualberta.ca \\
				CMPUT 391 & mario.nascimento@ualberta.ca \\
				PHYS 524 & gingrich@ualberta.ca \\
				PHYS 512 & andrzej.czarnecki@ualberta.ca \\
				INTD 350 & michael.bowling@ualberta.ca \\
				N344 & gosgnach@ualberta.ca
			\end{tabular}
			~\hspace{0.05cm}
			%    \caption{Academic staff}
			\begin{tabular}{l|l}
				{\bf Name} & {\bf Department}\\ \hline
				Rafiei, Davood & CS (2000) \\
				Nascimento, Mario A & CS (1999)\\
				Gingrich, Douglas M & Physics (1993) \\
				Prus-Czarnecki, Andrzej & Physics (2000) \\
				Bowling, Michael & CS (2003) \\
				Gosgnach, Simon & Physiology (2006)
			\end{tabular}
			~\hspace{0.05cm}
			%    \caption{White pages}
			\begin{tabular}{l|l}
				{\bf Name} & {\bf Phone} \\ \hline
				D Rafiei & (780) 433-6545 \\
				M A Nascimento & (780) 428-2108 \\
				D Gingrich & (780) 406-4565 \\
				A Prus-czarnecki & (780) 433-8303 \\
				M Bowling & (780) 471-0427 \\
				S Gosgnach & (780) 432-4814
			\end{tabular}
		\end{tabular}
	} % resizebox
	\caption{Example joinable tables}
	\label{fig:exampleTables}
	\shrink
\end{figure*}  

As an example, consider the tables shown in Figure~\ref{fig:exampleTables}. Two of the tables are obtained from the University of Alberta websites, and the third is looked up from Edmonton white pages with the last four digits of phone numbers changed for privacy reasons.
The two tables on the right list staff names, departments, and phones. The name column is common between the two tables, but the values are formatted differently, and the tables cannot be joined using equality join. The two tables on the left do not have any common columns, but they describe the same entities (here staffs), and one may come up with some rules to map names in one table to email addresses in the other table. Such mapping rules can be crucial when data is integrated from different sources, but constructing them is not always straightforward and can be quite time consuming. Instead, a tool may analyze the data sources and recommend mappings that
%join points and how data in one table can be mapped to
make a join possible. It is much easier for a data scientist to verify or select a subset of the recommendations made by the tool instead of coming up with all those mappings.

This paper studies the problem of joining tables despite the differences in data representation or formatting. In particular, we want to find transformations that map a column in one table to a column in another table with some desirable properties:
\begin{enumerate}
	\item \textit{efficiency} The space of possible transformations can be huge, especially for long textual columns and with a large number of rows. At the same time, a data scientist wants to learn about the mappings in real-time and make decisions on how the tables should be joined.
	\item \textit{noise handling} Data from different sources can be noisy and the matching rules can be complex. For example, no single rule can map names to email addresses, but a few rules may cover a large number of them.
	\item \textit{coverage} Transformations that have enough coverage may be obtained from a small subset of the data and are applied to a larger set, and this can be useful when data is changing and new rows are added over time. %We want the transformations that have enough coverage to be obtained from a small sample.
	\item \textit{optimality} 
	A table column may be mapped to another table column under many different sets of transformations, and some of those are more desirable than others. Ideally, we want to find the \textit{best} set of transformations, where the best may be defined in terms of the generality, the minimality, or other optimality criteria.
\end{enumerate}

Despite the large body of work on efficiently supporting equality join in relational databases (e.g., ~\cite{Shapiro1986join,Lee2012Join,Tian2016Join}), there are only a few that address the problem in the presence of mismatches, and a common theme here is to consider matching under a similarity function~\cite{Chaudhuri2003Robust,fastjoin}. A problem with a similarity join in general is that we know two rows are similar, but we may not know what makes them similar and how one can map one row to the other. 
Our work follows the line of work on program synthesis where we want to find declarative transformations under which an equality join becomes possible. 
Finding data transformations has a far greater impact than simply finding matching pairs with implications for predicting missing values~\cite{FlashFill,BlinkFill}, schema mapping~\cite{cate2017approximation,bonifati2019interactive} and data cleaning and repairing~\cite{he2016interactive}.
%	\note{Not only finding transformations allows us to perform join~\cite{autojoin}, but they can also be employed to fill missing rows~\cite{FlashFill,BlinkFill}, predict the values when target is not available, rectify textual errors~\cite{TBA}, and, as we will explain in this paper, detect and eliminate noise and outlier data instances. Moreover, transformation would provide a better controll over join and a data scientist may filter and optimize the transformations before joining tables.} 
In a recent relevant work, 
Zhu et al.~\cite{autojoin} find transformations that make one column joinable with another column under the assumption that a single transformation maps either all rows or a pre-selected subset of rows. 
Our work relaxes this condition since a single transformation may not cover all input rows (as can be seen in our example tables), and the user may not be able to select a subset that conforms to a single transformation. Compared to Zhu et al.~\cite{autojoin}, our algorithm is also faster by a few order of magnitude, which is achieved by leveraging the general characteristics of the matches and novel pruning strategies that are quite effective.
%We relax this condition since a single transformation may not cover all input rows (as can be seen in our example tables), and the user may not be able to select a subset that conforms to a single transformation. Our proposed method is also faster by a few order of magnitude, and this is achieved by developing a novel algorithm and some pruning strategies.

Our contributions can be summarized as follows:
%\begin{itemize}
%	\vspace*{-.35\baselineskip}
(1) we propose a novel framework for efficiently learning transformations based on the copying relationship between joinable row pairs;
(2) we develop a few strategies (e.g. early abandoning, caching, and bounding the parameter space) that significantly speeds up our transformation discovery, especially when there are millions of candidates;
%(2) We develop a pruning strategy based on early abandoning that significantly speeds up the coverage computation, especially when there are millions of candidate transformations.
(3) we provide an analysis of the performance of our algorithm in terms of the running time and the coverage of transformations, compared to a state-of-the-art approach. We also study the performance under sampling, as a way to scale our algorithm to larger data;
(4) we develop some optimality criteria for transformations, in terms of coverage and length, and an algorithm that finds the transformations under those criteria;
(5) we design synthetic data to show the scalability under different parameters including the number and the length of joinable row pairs. (6) Our code, implementing the baselines, synthetic data generator and real-world benchmarks, are made publicly available\footnote{https://github.com/arashdn/table-string-transformer}.
%\end{itemize}
%\vspace*{-.35\baselineskip}

The rest of the paper is organized as follows: we next provide a problem definition and some baseline methods that are applicable. Our proposed method is discussed in Section~\ref{sec:approach} and its performance and complexity is analyzed in Section~\ref{sec:comlexity}. Our experimental evaluation is reported in Section~\ref{sec:experiments}, and the related work is discussed in Section~\ref{sec:related_works}. Finally, Section~\ref{sec:conclusions} concludes the paper.

\section{Problem Definition}
\label{sec:problem_def}

We will introduce some terms and notations before presenting a formal definition of the problem.

\begin{definition}[\textbf{Transformation unit}]
	\label{def:transUnit}
	A transformation unit is a function that, when applied to an input string, copies either part of the input or a constant literal to the output.
	%A transformation unit is a simple function such that, when applied to an input string, identifies substrings with certain patterns and uses these substrings and the input parameters to construct an output string.
\end{definition}
\vspace*{-.15\baselineskip}

For simplicity, and when there is no confusion, \textit{transformation unit} may be referred to as \textit{unit} in the rest of this paper.
Our transformation units in this paper include the following string functions, which are commonly used in programming languages and are sufficient for many join scenarios. Clearly, the set can be expanded with additional units, and our algorithms should not be much affected as long they meet the requirements of a unit.

\begin{itemize}
	\item \textit{Substr(s,e)}  returns a substring of the input starting at position $s$ and ending at position $e$.
	
	\item \textit{Split(c, i)} splits the input using $c$ as the delimiter and returns the $i$\textsuperscript{th} string in the list.
	
	\item \textit{SplitSubstr(c, i, s, e)} splits the input using $c$ as the delimiter, takes the $i$\textsuperscript{th} string in the list, and returns the substring starting at position $s$ and ending at position $e$. This is equivalent to Split(c,i) followed by Substr(s,e).
	
	\item \textit{TwoCharSplitSubstr(c1, c2, i, s, e)} splits the input using $c1$ and $c2$ as delimiters, takes the $i$\textsuperscript{th} string in the list, and returns the substring starting at position $s$ and ending at position $e$.
	
	\item \textit{Literal(str)} returns $str$ irrespective of the input.
\end{itemize}

Our set of units includes all \textit{physical operators} in Auto-Join~\cite{autojoin} except SplitSplitSubstr(), which splits the input twice before taking a substring.

\begin{lemma}
	The two transformation units TwoCharSplitSubstr() and SplitSubstr() can express any transformation that can be expressed using the SplitSplitSubstr() of Auto-Join~\cite{autojoin}.
\end{lemma}
% May need package  "collection-fontsrecommended" for texlive
\begin{proof}
	The Substr action is the same for all three units mentioned in the lemma, and the mapping of its parameters is ignored. Let $c_1$ and $c_2$ denote the two split characters of SplitSplitSubstr(). There are five possible cases: (1) Neither $c_1$ nor $c_2$ exist in the input, and both SplitSplitSubstr() and SplitSubstr() are equivalent to Substr(). 
	(2) Only $c_1$ or $c_2$ occurs in the input, and SplitSplitSubstr() is equivalent to SplitSubstr().   
	(3) $c_1$ is followed by $c_2$ (i.e., the input conforms to $\Sigma^* c_1 \Sigma^* c_2 \Sigma^*$, where $\Sigma^*$ is any sequence of zero or more characters other than $c_1$ and $c_2$), and
	SplitSplitSubstr() may select a piece of text before $c_1$, after $c_2$ or between $c_1$ and $c_2$. The first two cases are covered using SplitSubstr() with either $c_1$ or $c_2$ as the split character, and the last case is covered by TwoCharSplitSubstr($c_1$, $c_2$, i, s, e) for some values of $i$, $s$ and $e$.
	(4) The input conforms to $\Sigma^* c_1 \Sigma^* c_1 \Sigma^* c_2 \Sigma*$, and SplitSplitSubstr() may select a piece of input before the first $c_1$, between the two $c_1$, between $c_1$ and $c_2$, or after $c_2$.
	The second case can be obtained using SplitSubstr() with $c_1$ and the other cases are covered above.
	(5) The input conforms to $\Sigma^* c_1 \Sigma^* c_2 \Sigma^* c_1 \Sigma^*$, and SplitSplitSubstr() may select a piece of input before the first $c_1$, after the last $c_1$, between $c_1$ and $c_2$, or between $c_2$ and $c_1$. The first case is obtained using SplitSubstr() with $c_1$, and the second and third cases are covered by TwoCharSplitSubstr(c1, c2, i, s, e) and TwoCharSplitSubstr(c2, c1, i, s, e) respectively for some values of $i$, $s$ and $e$.
\end{proof}

\begin{definition}[\textbf{Transformation}]
	A transformation $t$ is a sequence of units $t_1,t_2,\ldots$, and
	when $t$ is applied to a text $s$, denoted as $t(s)$, it produces       
	$t_1(s).t_2(s),\ldots$, i.e. the concatenation of the outputs of the units, each applied to $s$.
	
	%	A sequence of transformation units to be applied to the input is called a transformation.
\end{definition}

\begin{definition}[\textbf{Covering transformation set}]
	Given an input set $I$ of source and target pairs, a transformation set $T$ covers $I$ if for every input pair $(s,g) \in I$, there is a $\tau \in T$ such that $\tau(s)=g$. 
\end{definition}

The problem studied in this paper is to efficiently find transformations that make two tables equi-joinable. More specifically,
given two columns to be joined and a set of candidate row pairs from these columns~\footnote{Section~\ref{sec:end-to-end} discusses how candidate row pairs can be detected}, let $T$ be the set of possible transformations. Two columns may be joinable under one transformation from $T$ (when the mapping is more predicable) or a set of transformations. For example, phone number may be joined under one transformation whereas one may need a set of transformations to map name to email. This gives rise to two variations of the problem:
(1) there is a transformation $\tau \in T$ that covers the input or a large portion of it, and (2) the input is not covered by a single transformation but there is a set of of transformations that covers it.
For the former, we aim to find a transformation $\tau \in T$ that has the \textit{maximum coverage}. For the latter, we want to find a covering set of transformations $T_c \subseteq T$ that is concise, hence we aim for the \textit{minimal cover}.

However, with each transformation formed as a sequence of units, the number of transformations exponentially increases with the number of parameters of each unit as well as the number of units in the sequence; hence the set $T$ of possible transformations can be huge. Exhaustively searching for a maximum coverage or a minimal cover can be computationally expensive (if not intractable). Our objective is to find an efficient solution that scales well with the input size.

\section{Baselines}
\label{sec:baselines}
This section presents two baseline solutions before our approach is introduced in the next section. The first baseline is a naive approach, and the second one is a state-of-the-art method. In both approaches, it is assumed that the input is a set of source and target pairs, meaning the matching rows are identified. In  section~\ref{sec:ngram}, we discuss how those matching rows can be detected.

\subsection{Naive Approach}
A naive approach to the problem is brute force, which may be broken down into the stages (not necessarily sequential) of (1) enumerating the transformations and computing the coverage of each transformation, and (2) finding a transformation with the maximum coverage or a minimal covering set of transformations.

Each transformation is a sequence of units, and each unit can be, for example, any of the transformation units mentioned in section~\ref{sec:problem_def} with every possible parameter combination. For a pair of source and target strings, both of length $l$, the length of a transformation that can map the source to the target is bounded by $O(l)$. Hence the number of transformations in Stage 1 is exponential on the length of the strings. To compute the coverage of a transformation, one needs to apply it to all source and target pairs.
In Stage 2, finding a transformation with the largest coverage is straightforward, whereas finding a minimal covering set is NP-complete unless one resorts to a greedy approach.
The naive approach is considerably inefficient due to the large number of transformations, the time it takes to evaluate each transformation, and to find the covering sets.

\vspace*{-0.18\baselineskip}
\subsection{Auto-Join}
In a recent work by Zhu et al.~\cite{autojoin}, referred to as Auto-Join, the authors address the problem of explosion in the number of transformations by taking subsets of the input and finding a transformation for each subset. This approach hinges on the assumption that there is at least one transformation covering all pairs in each subset.
Based on this assumption, every transformation unit with all possible parameter combinations is applied to the pairs in each subset, and the units are sorted based on the length of the target text covered. Then, in a recursive process, the top unit is selected and is applied to the subset. Finally, any remaining text on the left and the right of the transformed text is treated as a new instance of the problem.
The method is recursively applied to the remaining text on the left and the right to find the best units on each side. The process continues until nothing is remaining on both sides or no transformation is found. In the former, the units from each side are attached together to form a transformation and in the latter, a rollback happens and the second best unit is selected and the search continues.

As an example, consider the tables in Figure~\ref{fig:exampleTables} and suppose this algorithm is applied to Rows 4-6 of the name columns of the two tables on the right. Ignoring the capitalization in text, we are looking for a transformation for
\textit{\{(``prus-czarnecki, andrzej'', ``a prus-czarnecki''), (``bowling, michael'', ``m bowling''), (``gosgnach, simon'', ``s gosgnach'')\}}.
A search over all transformation units reveals that
\texttt{Split(',', 1)} (i.e., split by ',' and choose the first item in the list) covers the maximum portion of the output. Applying this transformation leaves no text on the right of the expected target, and the remaining left will be
\textit{\{(``prus-czarnecki, andrzej'', ``a ''), (``bowling, michael'', ``m ''), (``gosgnach, simon'', ``s '')\}}.
Repeating the same process on the new set, both units \texttt{Literal(' ')} and \texttt{SplitSubstr('  ',2,0,1)} cover the maximum possible portion of the target and either of them can be selected. Applying \texttt{Literal(' ')} leaves no text on the right, and the remaining set on the left is:
{\{("prus-czarnecki,~andrzej", "a"), ("bowling, michael", "m"), ("gosgnach,~simon", "s")\}}. 
Finally, \texttt{SplitSubstr('  ',2,0,1)} covers all rows and the remaining text on both sides of the target will be empty, meaning the process ends. The algorithm returns with the following transformation:
%	\todo{Repeating the same process on the new set, both units \texttt{Literal(`` '')} and \texttt{SplitSubstr(``  '',2,0,1)} cover the maximum possible portion of the target.
	%	However, using \texttt{Literal(`` '')} there is still some remaining text in the target, whereas the unit \texttt{SplitSubstr(``  '',2,0,1)} covers all rows and the remaining text on both sides of the target will be empty, meaning that no further process in required. The algorithm returns the following transformation:}\\
{\small \texttt{<SplitSubstr('  ',2,0,1), Literal(' '), Split(',', 1)>}}.

This algorithm is based on the assumption that there is a transformation that covers all pairs in a subset. The crux of the algorithm is that a transformation can be instantly rejected when an input pair is hit that is not covered by the transformation. Since some subsets will not lead to any transformation, the authors apply their algorithm multiple times, each time on a different subset.
%\note{While in the paper, the authors rank these transformations to obtain the best one for the join, to fit it into our problem, we consider it as a set cover problem as described in section~\ref{sec:set_cover} to obtain a covering transformation set.}
%
The authors do not provide much hint on the size of a subset or the number of subsets, but it is easy to see that both directly depend on the coverage of transformations that are sought. For example, consider a transformation $t$ that covers half of the input pairs and a subset of size 5. The probability that all 5 pairs in a random subset will be covered by $t$ is $0.5^5=0.03125$, and one will need 32 subsets for the expected number of subsets covered by $t$ to reach one.

\section{Our Approach}
\label{sec:approach}
The key idea behind our approach is to guide the search process by exploiting any textual evidence for joinability. With a copying relationship among joinable row pairs, any piece of text that is copied from a source to a target provides some evidence on the shape of candidate transformations. We want to put together such evidence to eliminate transformations that cannot contribute to the final solution, hence, reduce the size of the search space.

In this section, we first present our approach for efficiently navigating the large space of transformations and piecing together solutions for the problems of maximum coverage and minimal cover. We then discuss how our approach can be integrated into an end-to-end system where candidate rows for join are not known.

\subsection{Discovering transformations}
\label{sec:approach_finding_trans}
Based on the concept of copying between joinable row pairs, the basic idea of our approach is to shrink the search space by considering only the transformations that can emit the given textual evidence (i.e., common sub-sequences among source and target pairs). In what follows, we first define \textit{placeholders} as a generalization of textual evidence and to serve as a seed to generate transformations. Next, we discuss the challenges in extracting transformations from placeholders and some strategies that can be employed to further shrink the search space with no significant loss in the quality of transformations.
\vspace{-.2\baselineskip}
\begin{definition}[\textbf{Placeholder}]
	Let $T$ denote the set of all nonconstant transformation units (i.e., the output is not the same for all input).
	Given a pair of source and target texts, a placeholder is any contiguous block of text in the target that can be obtained by applying a transformation unit in $T$ to the source.
\end{definition}

The significance of placeholders is that the search for transformations can focus on parts of the target text that can be obtained from the source instead of doing a blind search and potentially reducing the search space.
While the designation of a placeholder, in general, is dependent on the set of transformation units,
when a transformation unit is defined as in Definition~\ref{def:transUnit}, 
%for the set of transformation units discussed in Section~\ref{sec:problem_def},
it is easy to see that every non-constant transformation unit copies part of its input to the output. Hence, the set of placeholders includes every substring in the target, which is also a substring of the source.

\begin{figure}[t]
	\includegraphics[width=0.8\linewidth]{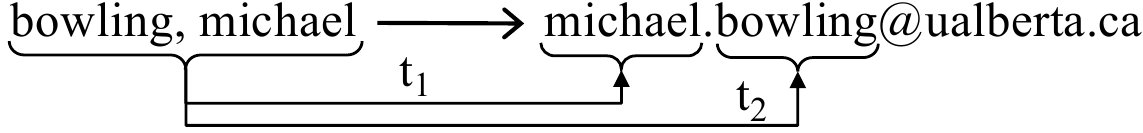}
	\centering
	\caption{Example source and target with two placeholders}
	\label{fig:transformation_example}
	
\end{figure}

For example, given the source ``bowling, michael'' and the target ``michael.bowling@ualberta.ca'', the blocks ``michael'' and ``bowling'' in the target are copied from the source and are considered placeholders. 
%With a set of placeholders detected, there must be at least one transformation unit that generates the placeholder. 
Let $t1$ and $t2$ be transformation units that map the source ``bowling, michael'' to placeholders ``michael'' and ``bowling'' respectively.
(as shown in Figure~\ref{fig:transformation_example}).
By finding candidate transformation units for each placeholder and concatenating them with literals, one will end up with a set of transformations that cover the input pair.
One such candidate transformation for our example is \texttt{<t1, Literal('.'), t2, Literal('@ualberta.ca')>}, where the concatenation of the outputs of the units in the sequence produces the target.

In the rest of this section, we first study some ways of bounding the search space by both grouping transformations that conform to a template or skeleton and introducing some notation of quality. Then we develop a few strategies and algorithms to efficiently enumerate the transformations to compute their coverage while keeping duplicates at bay.

\subsubsection{The space of transformations}
The main challenge in finding transformations is the search space.
The number of candidate transformations can be huge and searching the space of transformations to piece together a solution can be costly.
The number of candidate transformations increases with the number of rows, the number of placeholders per row, the length of placeholders\footnote{This is because every substring of a placeholder is also a placeholder.}, and the number of candidate transformation units for each placeholder.
Given a source and a target pair, by finding a transformation for each placeholder and concatenating them with the literals, one will end up with a transformation that covers the input pair. Nonetheless, many possible combinations of placeholders and literals  might be extracted. While some combinations may lead to transformations that cover a large number of pairs in the input, others may cover only a single row pair or very few pairs. To find a transformation with maximum coverage or a minimal covering set, one will need to evaluate all these candidate transformations.
Checking the coverage of each candidate transformation can also be costly when there are millions to billions of transformations, which was the case in some of our experiments (e.g. see Table~\ref{tab:cache}), and they need to be applied to all input rows.

Enumerating the set of transformations involves detecting combinations of placeholders. For a given input pair, based on the definition of a placeholder, each common $n$-gram among the source and the target for all values of $n$ can be considered as a placeholder. Hence all n-grams of all sizes in the target can be considered as placeholders, whereas
every block of target text that is not a substring of the source can be considered as a literal since no other transformation unit can generate it.

Each combination of placeholders and literals that generates the target of a row after a concatenation is referred to as a \textit{transformation skeleton} that fits the row. Each skeleton can give rise to many transformations.
While generating all possible skeletons is computationally expensive, a large portion of those combinations do not lead to desirable transformations and will not end up in the answer set.
We want to eliminate combinations that lead to transformations that are already covered by a \textit{better} transformations (e.g., covering more rows).

% Quality of transformations, Transformation filtering

\subsubsection{Transformation fitness}
%\textbf{Choice of transformations}
One question is if we can say a transformation better fits the input, and if a transformation can be ignored in favor of another transformation for the purpose of early filtering. Two measures of quality may be established. One measure is the transformation coverage. The more coverage a transformation has, the more general the transformation is. Consider two transformations $t1$ and $t2$ and suppose $t1$ covers every row covered by $t2$ plus at least one more. In any solution set of transformations that includes $t2$, the transformation can be replaced with $t1$ and the coverage cannot get worse (and it may get better). Hence  $t2$ may be removed in favor of $t1$.

Another measure of quality is the length of a transformation, in terms of the number of placeholders inside. A transformation with fewer placeholders is usually easier to read and may have less redundancy in capturing the structure in data.
This is also a desirable property that can reduce or avoid a possible fragmentation of placeholders.

\begin{lemma} [\textbf{Maximal-length placeholder principal}]
	\label{lemma:minimal-length}
	~\linebreak The minimum transformation length, in terms of the number of placeholders, is when the placeholders are of maximal length.
\end{lemma}

The maximal-length placeholder principal will not necessarily lead to transformations with larger coverage though. Consider the following two rows of source and target pairs.
\miniskip
\begin{align*}
	abcdefghijklmn &\rightarrow  defg.jkb  \\
	0123456789abcd &\rightarrow  d456.9ab
\end{align*}
Now consider the following three transformations:\\
\texttt{
	\small
	$t1$: <Substr(4,7), Literal('.'), Substr(9,10), Literal('b')>\\
	$t2$: <Literal('d'), Substr(5,7), Literal('.'), Substr(9,11)>\\
	$t3$: <Literal('d'), Substr(5,7), Literal('.'), Substr(9,10), Literal('b')>
}\\
%\begin{align*}
%	t1: &\textit{<Substr(4,7), Literal(``.''), Substr(9,10), Literal(``b'')>}\\
%	t2: &\textit{<Literal(``d''), Substr(5,7), Literal(``.''), Substr(9,11)>}\\
%	t3: &\textit{<Literal(``d''), Substr(5,7), Literal(``.''), Substr(9,10), Literal(``b'')>}
%\end{align*}
The first row is covered by $t1$, the second row is covered by $t2$, and both rows are covered by $t3$. All the placeholders in $t1$ and $t2$ are of maximal length whereas the placeholders in $t3$ are not of maximal length. $t3$  consists of 5 units instead of 4 in both $t1$ and $t2$.

% abcdefghijklmn --->  defg.jkb substr(4,4).'.',substr(9,2).'b'
% 0123456789abcd --->  d456.9ab 'd'.substr(5,3).'.'.substr(9,3)
%
% 'd'.substr(5,3).'.'.substr(9,2).'b'

%  With substr as unit, the maximum transformation coverage is when each placeholder is of length 1. (use this to show that the two objectives are contadictory)

\begin{lemma} [\textbf{Maximum transformation coverage principal}]
	\label{lemma:max-coverage}
	The maximum transformation coverage may not be reached with maximal-length placeholders and it can be reached with any placeholder length.
\end{lemma}
\begin{proof}
	The example given above shows, using Substr as a unit, that the maximal-length placeholder may not lead to the maximum coverage. Here we show using Split that the maximum transformation coverage can be reached at any placeholder length.
	Consider an input where each source row has a unique separator, and the target is the text on the right side of this separator. Since a split based on this unique separator generates the whole target, the transformation unit covers a maximal-length placeholder. On the other hand, each separator is unique to a row, and a split based on the row separator will only cover the same row and has a coverage of one. Now suppose the text on the right of the unique separator has some text that is common between multiple rows and can provide some common separators.
	Let the source be formatted as
	\textit{uniqueSep,nonCommonText, commonSep,nonCommonSuffix} and the target be \textit{nonCommonText, commonSep, nonCommonSuffix}. Any character in \textit{commonSep} may be used as a common separator for multiple rows. By changing the length of \textit{nonCommonText}, the maximum transformation coverage can be obtained at any placeholder length.
\end{proof}

As an example, consider the following two rows, showing the source on the left and the taget on the right of the arrows.
\begin{align*}
	12345sabcdefg &\rightarrow abcdefg \\
	67890taxxxx &\rightarrow axxxx
\end{align*}
Consider the two transformations \texttt{split('s',2)}, which covers the first row, and \texttt{split('t',2)}, which covers the second row. They both cover the maximal placeholders but have a coverage of only one row. It is easy to see that
\texttt{<Literal('a'),Split('a',2)>} has a coverage of two but the units are not covering maximal-length placeholders. For the same reason, the transformation length has increased.

\noindent
\subsubsection{Maximal-length placeholders as the backbone of transformations}
%\textbf{Search for maximal length placeholders}
Maximal length placeholders may be considered as the backbone of our transformations in that every desirable transformation either uses or is linked to maximal-length placeholders.
One way to reduce the size of the search space is to limit the placeholders to only those that are of maximal length. This not only reduces the number of candidate transformations significantly but also reduces the length of the transformations. A downside is that the maximal length placeholders may not lead to the maximum coverage (as shown in Lemma~\ref{lemma:max-coverage}) and some of those transformations may not be considered as candidates, but this is not difficult to address.

\begin{lemma}
	\label{lemma:placeholder_split}
	Given a transformation $t$ with all placeholders of a maximal-length, let $R_t$ denote the set of all input rows covered by $t$. Let $t^\prime$ be a transformation obtained from $t$ by replacing a maximal length placeholder $P$ with either two placeholders $P1$ and $P2$ and a literal $L$ or a placeholder $P1$ and a literal $L$, and denote the set of all rows covered by $t^\prime$ with $R_{t^\prime}$. $R_t$ can be a proper subset of $R_{t^\prime}$ (meaning $R_{t^\prime}$ can have more coverage than $R_t$) under two cases:\\
	(1) $P$ is broken to $P1\ L\ P2$, and $L$ is a common separator of $P1$ and $P2$ in all rows in $R_{t^\prime}$.\\
	(2) $P$ is broken to $P1\ L$ or $L\ P1$, and $P1$ is a placeholder in $R_{t^\prime} - R_t$ but $P$ is not a placeholder in $R_{t^\prime} - R_t$.
\end{lemma}

The intuition behind the first case is that a common separator falls inside a placeholder, and the maximal length placeholders cannot use the separator, whereas the non-maximal length placeholders can, and this gives rise to a better coverage, as discussed in the proof of Lemma~\ref{lemma:max-coverage}. 
One way to address this is to break maximal length placeholders based on separators that are expected to be common between multiple rows, and consider both the placeholders before and after the split as candidate placeholders.
Our experiments show that using only space and punctuations as possible common separators resolves all cases we have seen in our real datasets.
An example of our second case is shown in the example given before Lemma~\ref{lemma:max-coverage}.
The second case may be resolved by possibly combining transformations that are similar but cover different input sets if the combination provides a better coverage.

Given an input row, one may first obtain all maximal-length placeholders and their combinations with literals that cover the row. Then every maximal-length placeholder may be tokenized using common split characters in the natural language, such as punctuations and spaces, resulting in new skeletons (of placeholders and literals). As a result, we will have a set of skeletons that cover the row, and each skeleton can be used to generate a set of transformations (as discussed next).
For example, for the input pair (``Victor Robbie Kasumba'', ``Victor R. Kasumba''), the following skeleton set will  be produced:\\
\texttt{
	\small  
	\{<(P: 'Victor R'), (L: '. '), (P: 'Kasumba')>,\\
	\space\space <(P: 'Victor'), (L: ' '), (P: 'R'), (L: '. '), (P: 'Kasumba')>,\\
	\space\space <(L: 'Victor R. Kasumba')>\},
}\\     
where \texttt{P} indicates a placeholder and \texttt{L} identifies a literal.

\noindent
\subsubsection{Generating the transformations}
Generating candidate transformations from skeletons that cover a row is straightforward. Each placeholder can be replaced with a set of transformation units that map the source text to the text marked with the placeholder.
While a blind search for the parameters of a transformation unit can be computationally expensive, when the expected output (i.e., the placeholder text) and the matching part in the source are specified, the search is significantly faster. Consider a placeholder with text \textit{txt} and a matching source text that starts at position \textit{s} and ends at position \textit{e}. With the transformation units discussed in Section~\ref{sec:problem_def}, the placeholder can be replaced with 
(1) \textit{Substr(s,e)}, 
(2) \textit{Split(c,i)} where $c$ is a character at positions $s-1$ or $e+1$ of the source, $c$ does not occur in \textit{txt} and $i$ is an index that gives \textit{txt} after splitting the source, 
(3) \textit{SplitSubstr(c,i,s,e)} where $c$ is any character in the source that does not occur in \textit{txt} and $i$, $s$ and $t$ are possible indexes that generate the placeholder text,
(4) \textit{TwoCharSplitSubstr(c1,c2,i,s,e)} where $c1$ and $c2$ are any characters in the source that do not occur in \textit{txt} and $i$, $s$ and $t$ are possible indexes that generate the placeholder text, and
(5) \textit{Literal(txt)}. It can be noted that each placeholder may also be replaced with a literal, and this can be useful in cases where a constant in the target text occurs in the source by chance.

With each placeholder in a skeleton replaced with a set of candidate transformation units, a transformation is obtained by selecting a candidate unit for each placeholder.
Given an input row with a skeleton,
the Cartesian product of the candidate sets for all placeholders will give the set of candidate transformations. For example, consider the row (``Victor Robbie Kasumba'', ``Victor R. Kasumba'') and the skeleton $\text{\space} [$(P1: ``Victor R''), (L1: ``. ''), (P2: ``Kasumba'')$]$, and suppose the units that replace each placeholder are limited to Substr and Literal. The candidate units for $P1$ are 
\textit{\{Literal(``Victor R''), Substr(0,7)\}} and those for $P2$ are
\textit{\{Literal(``Kasumba''), Substr(14,21)\}}.
The Cartesian product of the extracted units would provide the following set of transformations:\\
\texttt{
	\small
	\{
	<Literal('Victor R'), Literal('. '),  Literal('Kasumba')>,\\
	<Literal('Victor R'), Literal('. '), Substr(14,21)>,\\
	<Substr(0,7), Literal('. '),  Literal('Kasumba')>,\\
	<Substr(0,7), Literal('. '),Substr(14,21)>
	\}
}

\subsubsection{Removing duplicates and computing the coverage}
The transformation generation phase can produce many candidates, and a large portion of them are duplicates.
Generally, the same transformation can be generated by multiple rows, and there is no need to keep more than one copy.
With the transformations stored in a hash set, duplicate transformations can be easily removed at the generation phase.

To compute the coverage of a transformation, one needs to apply it to all input rows and keep a record of the rows covered. This can be an expensive process when there is a huge set of transformations and they all need to be applied to all input rows. We address this problem by utilizing an eager filtering before each transformation is applied to a row. The filtering is at the level of transformation units and is very effective.
Consider a transformation $t$ that consists of a set of units and a row pair \textit{(src,tgt)}. $t$ cannot cover the row if the output of any of its units is not part of \textit{tgt}.
To speed up the computation, one may keep for each row a hash set of units that cannot be part of any transformation that covers the row. Before applying a transformation to a row, one can check if any of its units are present in the set of non-covering units of the row in $O(1)$ time using the row hash set. 
If a unit is present in the non-covering units of the row, the transformation can simply be ignored. Considering that the candidate set of transformations is a Cartesian product of candidate units, many units are repeated among the transformations and, as a result, such a filtering is very effective.

\subsubsection{Piecing together a final transformation set}
\label{sec:set_cover}
In the process of computing the coverage of transformations in the previous step, it is easy to keep track of the transformation(s) with the maximum coverage, or generally top $k$ transformations with the largest coverage.

The problem of finding a minimal covering set of transformations is the classic set cover problem, which is NP-complete. A greedy approach to the problem is to select in each step a transformation that covers the largest number of input rows that are not covered~\cite{clrs}. The greedy algorithm has an approximation ratio bound of $H(n) = \sum_{i=1}^n 1/i \leq ln(n) +1$, where $n$ is the maximum transformation coverage.

\subsection{An end-to-end join algorithm}
\label{sec:end-to-end}
Even though the focus of this paper is on finding a set of transformations that make a source column equi-joinable with a target column, in an end-to-end join algorithm, one will need to find joinable row pairs first. Also, two tables may join under not one, but multiple sets of transformations and an end-to-end join may be treated as a human-in-the-loop process.

\subsubsection{Finding joinable pairs}
\label{sec:ngram}
Joinable pairs often represent the same entities but may be formatted or described differently. The algorithm discussed in Section~\ref{sec:approach_finding_trans} can find a mapping to transform one formatting to another formatting, but it is assumed that the joinable pairs are given. Analogous to training data in machine learning, the joinable pairs may be tagged in advance and be provided as input. The tagged row pairs can be a small subset of the row pairs being joined and still represent the mapping relationships.

An alternative is to automatically identify row pairs in the source and target columns that may join. Several approaches have been proposed to address this problem in the literature based on some form of string similarity (e.g., edit distance, fuzzy token matching)~\cite{auto-em,MassJoin} or pre-trained models~\cite{fastjoin}.
Since our transformation units are in the form of string operations, finding semantically joinable rows or those labeled as relevant using a knowledge base but not syntactically similar are less beneficial and may even adversely affect the process of finding transformations. As a result, we employ an n-gram matching method that is able to retrieve joinable row pairs based on their textual similarities.

Since placeholders are the backbone of our transformations, joinable row pairs are expected to have some n-grams in common. A simple approach is to consider the joinable rows as those which have at least one n-gram in common. However, using one common n-gram as the join indicator may retrieve many false positives due to the presence of stop words, common prefixes (e.g., ``Dr.'', ``Professor'', ``Gov.''), common suffixes, etc.
For example, in Figure~\ref{fig:exampleTables}, if a person name contains ``albert'', it will match all email addresses in our course-contact table. To overcome this, we want to be more selective in our choice of n-grams and possibly find some entity descriptions, in terms of n-grams, that are maybe unique to the entity. In the spirit of Inverse Document Frequency (IDF) in IR, we define Inverse Row Frequency (IRF) for a token or an n-gram $t$ in column $c$ as follows:
\begin{equation}
	IRF(t, c) = \frac{1}{\text{Number of rows in } c \text{ that contain } t}.
\end{equation}
We define the representative score (Rscore) of a token or an n-gram $t$ that appears in both source column $SC$ and a target column $TC$ as
\vspace{-0.25\baselineskip}
\begin{equation}
	Rscore(t) = IRF(t, SC)  .  IRF(t, TC).
\end{equation}
\vspace{-0.1\baselineskip}
Although the scoring function is symmetric, our algorithm for finding joinable row pairs is not symmetric and distinguishes between source and target columns. When the source and target columns are not specified in advance, we may tag the more informative column with more description as the source column. Since our approach is based on textual descriptions, a simple approach is to consider the column with longer descriptions on average as the more informative and, accordingly, the source column.

With the source and target columns tagged, we want to find some representative n-grams for each row in the source. Since a single n-gram size does not work for all rows, we select for each source row and each n-gram size, for $n_0 \leq n \leq n_{max}$, an n-gram with the largest Rscore as the representative n-gram of that size. For a given source row $s$, we say a target row $t$ is a potential candidate for join if $t$ contains at least one representative n-gram of $s$.
A join between source and target columns can be one-to-one, one-to-many, many-to-one and many-to-many. Unless this relationship is explicitly specified, we assume the relationships can be many-to-many, and
if a row in the source column matches more than one row in the target column, all matching pairs are considered as separate candidate pairs. 
Algorithm \ref{alg:row_matching} describes this process in details.  

\tinyskip
\begin{algorithm}
	\SetKw{In}{in}
	\SetKw{To}{to}
	\SetKw{From}{from}
	\SetKwInOut{Input}{input}
	\SetKwInOut{Output}{output}
	\SetKwData{row}{row}
	\SetKwData{token}{token}
	
	\Input{Source column $SC$ and  target column $TC$}
	\Output{Pairs of candidate joinable rows}
	
	\BlankLine
	\For {$n$  \From $n_0$ \To $n_{max}$}
	{
			\ForEach{\row $r$  \In $SC$}
			{
					T = all n-grams (of length $n$) in $r$\\
					rep = $\underset{t \in T}{\mathrm{argmax}} \text{\space}Rscore (t)$\\
					C = all rows in $TC$ that contain rep.\\
					\ForEach{\row $r'$  \In $C$}
					{
							\lIf{$(r,r')$ $\notin$ output}{add pair $(r,r')$ to  output}
							
						}
				}
		}
	\caption{Pseudocode for finding candidate joinable pairs}
	\label{alg:row_matching}
\end{algorithm}
\shrink

To speed up the process of finding joinable row pairs, we build an inverted index for n-grams that appear in either the source or the target columns. For a fast access, the inverted index is organized as a hash with every n-gram of size $n_0 \leq n \leq n_{max}$ as a key and the row ids where the n-gram appears as a data value. For a source row of length $L$, the set of joinable rows can be obtained in $O(L)$ lookup.

\section{Performance Analysis}
\label{sec:comlexity}
A challenge in transforming tables for joinability is the large search space and the time it takes to search this space. This section analyzes the running time of our algorithm in terms of the input size, the length of the rows being joined, the length of transformations in terms of the number of placeholders, and the number of units. We compare the running time to that of Auto-Join, a state-of-the-art method in the literature.

\subsection{Running time of our approach}
Our approach for transforming a table for joinability consists of the following steps: (1) finding joinable row pairs, (2) detecting placeholders, (3) constructing transformation skeletons, (4) generating transformations, and (5) finding the coverage and compiling a solution.

With an inverted index on all n-grams in source and target columns, finding joinable row pairs is straightforward. For each source row, the representative n-grams can be selected and all target rows that contain those n-grams can be looked up in $O(1)$ time, assuming the inverted index is organized as a hash with O(1) access time and that each source row joins with a constant number of target rows.

For a given source and target pair, each n-gram in the target can be a placeholder. If $l$ denotes the length of a row (either source or target), the number of maximal-length placeholders is at most $l$. Each placeholder in the target can have at most $l$ matches in the source. Hence, all placeholders and their matches can be found in $O(l^2)$ time. However, the number of placeholders is generally much smaller than $l$, especially when we are interested in maximal-length placeholders only. Let $p \leq l$ denote the number of placeholders per transformation. Each placeholder can have up to $l$ matches in the source, giving rise to $pl$ placeholder combinations per row.

Each n-gram in the target that appears in the source can be considered both a placeholder and a literal, to capture the cases where a literal has a match in the source by chance. Hence the number of transformation skeletons is bounded by $2^ppl$.
Our experiments on real-world data show that only in very few cases a placeholder will match with more than one token in the input, and the number of skeletons  per row is bounded by $2^p$.

Let $u$ denote the number of available transformation units ($u$ is 5 for the set of units discussed in Section~\ref{sec:approach_finding_trans}). For some units, an expected output is generated with a single assignment of the parameters (e.g., Literal), while for others, it can be generated with multiple assignments. With the number of assignments per unit bounded by $l$, each placeholder can be replaced with $ul$ instances of units and parameter assignments.
Each skeleton can have up to $p$ placeholders, and each placeholder can be replaced by $p.u.l$ instances of units and parameter assignments.              
Hence, the number of candidate transformations per skeleton is $(ul)^p$ (i.e., the size of the Cartesian product). 
With $2^p$ skeletons per row, all transformations of a row can be found in $O((2ul)^p)$ time.

To compute the coverage of a transformation, one will need to apply it to all rows. On an input with $n$ rows, this can be done in $O(n)$ time per transformation, assuming $p$ is a constant.
Our greedy algorithm for finding a minimal covering set runs in linear time in the sum of the coverage of the transformations. With the number of transformations bounded by $n(2ul)^p$ and the coverage of each transformation bounded by $n$, a minimal covering set can be obtained in $O(n^2(2ul)^p)$ time.
With $u$ and $p$ treated as  constants, our algorithm runs in $O(n^2l^p)$ time.

It is worth mentioning that our cost analysis does not take into account our filtering and duplicate removal, which play an important role in reducing the size of the search space and speeding up our algorithm. We evaluate the effectiveness of those pruning steps in our experiments, as reported in Section~\ref{sec:pruning_performance}. We also assume that all of the transformation units have a parameter space of $l$. However, the parameter space for many units (such as Literal, Split, and Substr), when replacing a specific placeholder, is $O(1)$, and this can have a significant impact on the overall complexity of the method. For example, with this assumption, each placeholder can be replaced with $O(u)$ instances of transformations and parameters, and our algorithm runs in $O(n^2)$ time.

\subsection{Compared to Auto-Join}
To better assess the running time of our approach and to have a baseline for comparison, we provide a quantitative analysis of Auto-Join.

\noindent
\subsubsection*{Running time of Auto-Join}
Auto-Join selects a subset of the input and aims at finding a transformation that covers all rows in the subset. Let $r$ denote the size of this subset. The algorithm for finding a transformation is a recursive process where at each step, one unit with a parameter assignment is selected and a transformation around that unit is built within the recursive calls (if possible). Let $u$ denote the number of transformation units, and $z$ be the number of parameters per unit (e.g., z is two for Substr and five for TwoCharSplitSubstr). On an input pair of length $l$ (i.e., either source or target length), there are $ul^z$ choices of units and their parameter assignments. For each choice, one needs to apply the unit to all rows in the subset and detect for each row the part of the target that is covered. This can be done in $O(rul^{z+1})$ time.
All units that cover the subset are candidates for forming a transformation. These candidates are sorted based on the average length of the target they cover, and are considered in that order.
Sorting can be done via a heap with a very small overhead and is ignored in our calculation.

In each step of the recursion, a unit with an assignment of its parameters is selected (out of $ul^{z}$ choices) and the algorithm is recursively called to transform the remaining target text on both the left and the right of the matching text.        
If the transformation does not cover all rows of the subset, the algorithm backtracks and selects the next unit, in the sorted order, until a transformation is found or all units are tried.                  

Each unit that covers the input must transform at least one character in the target, leaving $l-1$ characters to be covered in the next recursive call. If the character that is covered always falls in the middle, the process is called for the left and the right portions of the target, each of length $l/2$, and the cost can be written as
$
C(l)= rul^{z+1} + 2ul^{z} C(\frac{l}{2}).
$
The number of recursive calls is bounded by $log l$ and the number of placeholders in a transformation, denoted by $p$. With $p \leq log l$ the algorithm runs in 
$O(2^{p -1} u^p l^{z p} 2^{-(p + 1)/2}(rl+2))$ time.
Assuming that both the number of placeholders and $u$ are constants, the algorithm runs in $O(l^{z p +1}r)$ time.
If the character that is covered always falls at the end, the process is called for the remaining $l-1$ characters, and the cost can be written as
$
C(l)= rul^{z+1} + ul^{z} C(l-1).
$
The number of recursive calls is bounded by $l$ and the algorithm runs in    
$O(u^{l} (l!)^{z} (rl+1))$ time. Again if the length of the recursion is bounded by the number of placeholders $p$ and that both the number of placeholders and $u$ are treated as constants, the algorithm runs in $O(l^{z p +1} r)$ time.

A subset selected by Auto-Join may not yield a transformation, or the transformation obtained may not have high coverage on the whole input. To avoid this, Auto-Join has to run on multiple subsets, and this should be considered in the cost estimate.

\noindent
\subsubsection*{A comparison}
Some observations can be made in comparing our algorithm with Auto-Join. First,
the running time of both our approach and Auto-Join are affected by the input length, but the running time of Auto-Join grows at a much faster pace with the exponent $zp+1$ instead of $p$ in our approach. 
Second, our analysis gives the worst case running time, and this case commonly happens in Auto-Join simply because one single noisy pair in the subset will push Auto-Join to search the entire parameter space, while the worst case is extremely rare in our approach. 
Third, our approach uses pruning strategies that significantly boost the performance. This is not shown in our cost analysis but is discussed in our experiments.
Finally, our algorithm runs once on the whole input, whereas Auto-Join runs multiple times on different subsets and our cost estimates for Auto-Join should be multiplied by the number of subsets tried.

\subsection{Performance under sampling}	\label{sec:sampling}
With our algorithm being quadratic on the input size,
one way to scale our algorithm to large input sizes while keeping the running time under control is sampling. A question is how small a sample can be and what transformations are discovered or missed under sampling.

Consider a transformation $t$ and a random sample of size $s$, and let $q$ denote the coverage of $t$ in terms of the fraction of input that is covered. In other words, $q$ is the probability that an arbitrary input row pair is covered by $t$. If $n$ denotes the number of input pairs, the probability that $t$ does not cover any row in the sample is $P_0 = (1 - q)^s$, and $1-P_0$ gives the probability that $t$ covers at least one row.
However, being covered by only one row in the sample is not sufficient to discover a transformation. A transformation that contains only a literal will cover a row, but no other rows, and one such transformation will be less useful. We need at least two rows in the sample to cover $t$. 
The probability that $t$ covers only one input row in the sample is
\vspace*{-0.2\baselineskip}
\begin{equation*}
	P_1 = {s\choose 1} \text{\space} q  \text{\space}  (1-q)^{(s-1)},
\end{equation*}
\vspace*{-0.1\baselineskip}
and $1 - P_0 - P_1$ gives the probability that at least two rows in the sample support $t$.
It is easy to see that even a relatively small sample provides enough data for our approach to discover a transformation. 
For instance, consider a relatively large input and a transformation that is covered by at least $5\%$ of the rows. In a sample of size 100, the probability that the transformation is discovered by our approach is $0.96$, which indicates that even a transformation with such a low coverage is discovered using a small sample. 

For a comparison, Auto-Join also does sampling with its selection of subsets, but its sampling is slightly different. Auto-Join requires all rows in a subset to be covered by a single transformation, and the probability of a subset being covered by a single transformation is $q^s$. 
This probability takes its maximum when the sample size is $1$, but
selecting $s=1$ can end up choosing a literal as the transformation. Hence at least two rows are required in a sample. 
Auto-Join also takes multiple subsets to increase the chance of discovering a transformation in at least one subset.
With $k$ different subsets, the expected number of subsets that cover the transformation is $k q^s$.
Consider the same example where a transformation that is covered by at least $5\%$ of the rows is sought. To have a subset that covers the transformation (i.e., an expectation of one), Auto-Join will need at least 400 different subsets.

\section{Experiments}
\label{sec:experiments}
This section reports an experimental evaluation of our algorithms and pruning strategies on both real and synthetic data and under different parameter settings.

\subsection{Dataset}
\label{sec:dataset}
Our evaluation is conducted on three real datasets: (1) a set of web tables with joinable tables paired up, (2) a collection of table pairs containing examples of common data cleaning tasks faced by spreadsheet users, and (3) open government data joined with data from non-government sources. We also evaluate our work using synthetic data where tables are generated for joinability under different parameter settings.

\noindent	
\textbf{Web dataset}
The web dataset is a benchmark introduced in Zhu et al.~\cite{autojoin}. The tables were collected by sampling table-content queries from a query log (e.g. ``list of California governors''), searching Google Fusion tables for those queries and selecting pairs of tables from the search results where data is formatted differently but the tables are joinable under some transformations.
The dataset consists of 31 web table pairs, covering 17 various topics. Each table has on average 92.13 rows, and the average length of a join entry is 31 characters. This is considered a difficult benchmark due to inconsistencies in data and differences in entity representation that may not be resolved using transformations.

\noindent
\textbf{Spreadsheet dataset}
This dataset, published in Syntax-Guided Synthesis Competition (SyGuS-Comp) 2016~\cite{sygus2016}, includes
the public benchmarks of both FlashFill~\cite{FlashFill,FlashFill2} and BlinkFill~\cite{BlinkFill}. The dataset contains 108 table pairs, collected from Microsoft Excel product team and help forums, and has many examples of common data cleaning tasks of spreadsheet users. Each table has on average 34.43 rows. 

\noindent	
\textbf{Open Governmental Dataset}
This dataset includes table pairs where one table is obtained from open government data and another table from non-government sources. Our government data included 
close to 3 million property assessments from the city of Edmonton~\footnote{https://data.edmonton.ca} and our non-government data included a random sample of people and businesses listings from the Canadian white pages~\footnote{https://whitepagescanada.ca}. %The former contains 2,975,437 rows and the latter includes 3808 rows. 
The two tables were joinable on the address field. To create a golden set of matching rows, we manually developed some rules (in the form of regular expressions) that gave us an initial set of joinable candidates, and the final joinable rows were manually picked and validated.

\noindent	
\textbf{Synthetic dataset}
This is a generated set of tables to evaluate the scalability and the performance of our algorithms in a more controlled setting. Our synthetic data was in the form of pairs of tables where each row in one table joins with a row in the other table under some transformation.
Synth-N refers to a set of tables where each table has N rows, and the length of each row in the source table is randomly chosen in the range $[20,35]$. Synth-NL refers to a set of tables with longer rows where again each table has N rows, but the length of each row in the source table is randomly chosen in the range $[40,70]$. For example, Synth-50 and Synth-50L refer to tables with 50 rows, each row of length in the ranges $[20,35]$ and $[40,70]$ respectively.
The generation starts by creating a source table where each row is an alphanumeric string of a random length in the specified range. Then a set of transformations are generated for each source table. Each transformation consists of $p$ placeholders and $l$ literals, randomly chosen from the set of possible units each with a random set of valid parameters, and placed in a sequence to form a transformation. In our experiments, $p$ was set to 2, $l$ was chosen randomly from $\{1,2\}$, the length of a literal block was in range $[1,5]$, and the number of transformations to cover a source table was set at 3. Once the transformations were fixed for a source table, a transformation was randomly chosen and were applied to each source row to generate a target row.

\subsection{Experimental Setup}
\label{sec:exp_setup}
% transformation length
One parameter that controls the length of the transformations, as well as the size of the search space, is the number of placeholders or the tree depth in Auto-join. This parameter is a trade-off between the coverage and runtime of the method and is set to 3 in our experiments on web tables, open data, and synthetic datasets and 4 on the spreadsheet data due to existence of more smaller textual pieces in this dataset.
Our transformation units were those reported in Section~\ref{sec:problem_def} except \textit{TwoCharSplitSubStr}, which was excluded to better manage the runtime, especially for our baseline which struggled in some of our datasets. This did not have much impact on our results though since the transformation was less common in our datasets.
%To further manage the runtime, especially for our baseline that struggles in some of our datasets, we exclude \textit{TwoCharSplitSubStr} in the set of units that form a transformation. This unit is less likely to arise in our real dataset and is excluded in our synthetic data generation, hence it does not impact our results. Otherwise, running our baseline with our available resources was not feasible.

% parameters of auto-join
Auto-Join takes as parameters the number of samples (referred to as subsets) and the size of a sample. The approach is sensitive to these parameters, but the authors do not provide much guidance on how they should be set~\cite{autojoin}. Based on our analysis in Section~\ref{sec:sampling} and after experimenting with different sample sizes, setting the sample size at 2 yields the maximum coverage and the best pattern coverage in our benchmark dataset, and this is how this parameter is set in our experiments. The number of samples is set to 6 to keep the approach executable with our resources. The transformations extracted by all these samples form up the final covering set.

% ngram size
Our row matching to find joinable row pairs is done based on representative n-grams of sizes $[n_0,n_{max}]$. In our experiments on the benchmark dataset, $n_0 = 4$ yields the best f-score, hence $n_0$ is set to 4 in all our experiments. The value of $n_{max}$ is set to 20, which is large enough to generate n-grams of roughly up to half the length of the input rows.

% number of trials and the machine setup
Unless explicitly stated otherwise, all our performance result on the web and spreadsheet datasets are the mean over all tables of the dataset, and on the synthetic data is the mean over 10 independently generated tables with the same parameters. Our experiments are conducted on a machine equipped with AMD EPYC 7601 processor and up to 64~GB of memory allowed for each experiment.

\subsection{Effectiveness of Row Matching}

	\begin{table}[tbp]
	\centering
	\caption{Row matching performance}
	\shrink
	\begin{tabular}{lcccccc}
		\toprule
		\multicolumn{1}{c}{Dataset} & \#Rows & Avg Len. & \#Pairs & P     & R     & F1 \\
		\midrule
		Web tables & 92.13 & 31.08 & 112.55 & 0.81  & 0.93  & 0.86 \\
		Spreadsheet & 34.43 & 18.59 & 32.44 & 0.95  & 0.93  & 0.94 \\  
		Open data & 3808 & 19.33 & 360,125 & 0.01  & 0.92  & 0.02 \\
		Synth-50 & 50    & 27.59 & 44.20 & 1.00  & 0.88  & 0.94 \\
		Synth-50L & 50    & 55.41 & 48.00 & 1.00  & 0.96  & 0.98 \\
		Synth-500 & 500   & 27.64 & 416.10 & 0.97  & 0.81  & 0.87 \\
		Synth-500L & 500   & 55.26 & 460.40 & 0.96  & 0.89  & 0.92 \\
		\bottomrule
	\end{tabular}%
	%		\shrink
%	\miniskip
	\label{tab:row_matching}%
\end{table}%

Table~\ref{tab:row_matching} summarizes the performance of our n-gram row matching in terms of precision, recall, and F1-measure. There is a one-to-one matching relationship between source and target rows in our datasets, and a perfect number of matching pairs should be equal to the number of input rows.
Our row matching achieves an average precision and recall higher than 0.8 on the web and spreadsheet benchmarks, and all our synthetic datasets. Open data poses an interesting challenge with many addresses wrongly matched, leading to a recall of 0.92 but a precision of only 0.01. Generally maintaining a high precision in row matching is important in reducing both the number of bogus transformations and the transformation discovery time. That said, we will see in the next section that our transformation discovery algorithm nicely recovers from a poor row matching in Open data using sampling as discussed in Section~\ref{sec:sampling} and by applying a support threshold on transformations.
A high recall in row matching can also be important in cases where an input table or sample does not have many rows, and some transformations may not garner enough support otherwise. Increasing the input length does not much affect the performance, but the expected number of matching rows increases with the numbers of input rows.

\subsection{Transformation Coverage and Runtime}

\begin{table}[tbp]
	\setlength{\tabcolsep}{1.4pt}
	\centering
	\caption{Comparison of the performance and runtime  of our approach and the baseline}
	\shrink
	\begin{tabular}{|r|l||c|c|c|c|}
		\toprule
		\multirow{1}[4]{*}{\begin{sideways}\textbf{\scalebox{.55}{Matching}}\end{sideways}} & \multirow{2}[4]{*}{\textbf{Dataset}} & \multicolumn{4}{c|}{\textbf{Our Approach (Auto-Join) Compared}} \\
		\cmidrule{3-6}          &       & \textbf{Top Cov.} & \textbf{Coverage} & \textbf{\#Trans.} & \multicolumn{1}{c|}{\textbf{Time (Sec)}} \\
		\midrule
		\midrule
		\multirow{7}[2]{*}{\begin{sideways}\textbf{N-Gram}\end{sideways}} & Web tables & 0.58 (0.39) & 1.00 (0.43) & 25.71 (2.65) & 22 (269,174) \\
		& Spreadsheet & 0.73 (0.66) & 1.00 (0.77) & 6.32 (2.81) & 11 (77,390) \\
		& Open data & 0.30 (0.00) & 0.56 (0.00) & 3.00 (0.00) & 23386 (91177) \\
		& Synth-50 & 0.42 (0.42) & 1.00 (0.42) & 3.00 (1.00) & 5 (84,463) \\
		& Synth-50L & 0.40 (-) & 1.00 (-) & 3.00 (-) & 21 ($>$650,000) \\
		& Synth-500 & 0.39 (0.39) & 1.00 (0.71) & 18.00 (3.00) & 232(239,559) \\
		& Synth-500L & 0.35 (-) & 0.68 (-) & 49.00 (-) & 1026 ($>$650,000) \\
		\midrule
		\midrule
		\multirow{7}[2]{*}{\begin{sideways}\textbf{Golden}\end{sideways}} & Web tables & 0.58 (0.37) & 1.00 (0.44) & 13.94 (3.13) & 7 (200,281) \\
		& Spreadsheet & 0.78 (0.72) & 1.00 (0.84) & 6.00 (2.77) & 10 (52,819) \\
		& Open data & 0.30 (0.15) & 0.66 (0.15) & 8.00 (1.00) & 4147 (124,626) \\
		& Synth-50 & 0.42 (0.42) & 1.00 (0.42) & 3.00 (1.00) & 6 (302,647) \\
		& Synth-50L & 0.40 (-) & 1.00 (-) & 3.00 (-) & 27 ($>$650,000) \\
		& Synth-500 & 0.39 (-) & 1.00 (-) & 3.00 (-) & 432 ($>$650,000) \\
		& Synth-500L & 0.35 (-) & 1.00 (-) & 3.00 (-) & 2119 ($>$650,000) \\
		\bottomrule
	\end{tabular}%
	\label{tab:results}%
	\shrink
	\miniskip
\end{table}%

% what the tables show, sampling on open data
Table~\ref{tab:results} shows the performance of our approach and the baseline (shown inside parenthesis) in terms of the coverage of the transformations that are discovered, the number of transformations in the covering set, and the running time in seconds. The coverage is shown in terms of both the coverage of a single transformation that has the highest coverage and that of the covering set, respectively referred to as \textit{Top Cov.} and \textit{Coverage} in the table.
As noted in the previous section, an n-gram based row matching on Open data produces a large number of matching pairs and 99\% of those pairs are false matches. To reduce the input size for our transformation discovery, we sample this data by randomly taking 3000 pairs (i.e. a sampling rate of less than 1\%). At the same time, those false matches give rise to a large number of transformations with low support, but it is easy to get rid of those transformations by setting a minimum support threshold.
For our Open data, we set a support threshold of 1\% but no support threshold was set on our other datasets.

% the issues with auto-join
Auto-Join really struggles on some of our tables and does not finish even within a week. We set a time limit of 650,000 seconds (which is roughly 180 hours or one week), and if the algorithm does not finish within this time limit, we set its coverage to 0 and its time to 650,000 seconds. On the web dataset of 31 tables, this happens for 8 tables under the golden row matching and 11 tables under our n-gram row matching. The spreadsheet benchmark is less noisy, and the time limit is reached for 8 tables with golden row matching and 12 tables under n-gram matching out of 108 table pairs in the dataset.

% results in top panel and bottom panels
In the top panel, the joinable row pairs are identified using our n-gram matching, and the algorithms for finding transformations are applied on the result of the row matching. The coverage of our approach in a few cases is less than 1.00. This happen in open data due to sampling and using a support threshold and in some of other datasets dues to a not-perfect performance of row matching.
With an imperfect row matching, some transformations can be missed or additional transformations may be generated (e.g., as shown for Synth-500 and Synth-500L where the number if transformations is larger than 3).
In the bottom panel, the result is under a gold standard row matching where joinable row pairs are given and our algorithm is applied on those rows. The coverage of our approach on all datasets is 1.00 except on open data where a support threshold was set on transformations.

A few observations can be made about the results. First, in our sampling of open data, we took only 3000 pairs out of about 360,000 matched pairs identified by n-gram matching. Interestingly, the sample included less than 1\% of the data, but the top transformation coverage was identical to that under a golden row matching. Also the coverage of top three transformations was very close, which stood at 0.56 when rows were identified using n-gram row matching and 0.59 under a perfect row matching. This indicates that even a small sample can yield a high coverage and our approach can be effectively adapted on datasets with millions of rows.
Second, comparing our approach to auto-join, we can see that our algorithm performs comparable or better in finding the top covering transformation with a running time that is 3-4 orders of magnitude better. While auto-join exhaustively searches the transformation space, it still can miss some transformations due to its sampling and the presence of noise in the input, whereas our approach is more robust (see our discussion in Section~\ref{sec:sampling}). On the problem of finding a covering set, Auto-Join does not find a covering set and only returns all transformations that cover at least one subset. For a covering set, we took all those transformations returned by auto-join. As shown in the table, the coverage of auto-join is, except in one case, all under 0.45, whereas our algorithm achieves a coverage of 1.00 except one case under n-gram matching and another case due to a limit on transformation support.
As expected, the running time of both approaches increase with the input length, but the increase in running time is only modest for our approach.

\subsection{End-to-end Join performance}
\begin{table}[tbp]
	\setlength{\tabcolsep}{1.1pt}
	\centering
	\caption{Join performance of our approach and baselines}
	\shrink
	\begin{tabular}{|l||c|c|c||c|c|c||c|c|c|}
		\toprule
		& \multicolumn{3}{c||}{Our Approach} & \multicolumn{3}{c||}{Auto-FuzzyJoin} & \multicolumn{3}{c|}{Auto-Join} \\
		\midrule
		Dataset & P     & R     & F     & P     & R     & F     & P     & R     & F \\
		\midrule
		Web tables & 0.879 & 0.726 & 0.713 & 0.935 & 0.672 & 0.708 & 0.985 & 0.415 & 0.466 \\
		Spreadsheet & 0.995 & 0.792 & 0.812 & 0.943 & 0.662 & 0.691 & 0.997 & 0.768 & 0.796 \\
		Open data & 0.955 & 0.553 & 0.700 & 0.601 & 0.700 & 0.647 & -     & -     & - \\
		Synth-50 & 1.000 & 0.964 & 0.979 & 0.977 & 0.512 & 0.657 & 1.000 & 0.420 & 0.592 \\
		Synth-50L & 1.000 & 0.998 & 0.999 & 0.985 & 0.656 & 0.755 & -     & -     & - \\
		Synth-500 & 1.000 & 0.831 & 0.890 & 0.960 & 0.455 & 0.584 & 1.000 & 0.712 & 0.832 \\
		Synth-500L & 0.995 & 0.929 & 0.955 & 0.986 & 0.642 & 0.750 & -     & -     & - \\
		\bottomrule
	\end{tabular}%
%	\shrink
	\label{tab:join}%
\end{table}%

An important area of application for the generated transformations is joining two tables. To evaluate the end-to-end performance of transformations, we build a platform to perform the join both in our approach and Auto-join. We apply the transformations with a minimum support (set to 2\% for open data and 5\% for all other datasets) on the source column and for any rows with a matching transformed value on the target, we perform an equi-join to obtain joinable rows. We compare the performance to both Auto-join and another state-of-the-art method proposed by Li~et al.~\cite{AFJ:2021} referred to as Auto-FuzzyJoin (AFJ). AFJ does not use transformations for joining, and instead, it considers various similarity measures and heuristics to generate a boundary for joinable and non-joinable rows. As shown in Table~\ref{tab:join}, our approach outperforms both baselines in terms of F1-Score on all datasets. AFJ mainly uses similarity functions to detect joinable pairs; it does not return any transformations, and it cannot provide interpretable join patterns as our approach. Also, AFJ cannot perform well in cases where the source column is not a key and may contain duplicate values or noisy data such as our open data benchmark. Since Auto-join only finds transformations that cover all rows in a subset, it is able to achieve a higher precision. However, for the same reason, Auto-join misses many transformations for join, especially, in noisy data such as web tables, and it has lower recall and F1-Score. The precision-recall trade-off in our approach can be adjusted by the choice of a minimum support for the transformations.

\subsection{Effectiveness of Pruning Strategies}
\label{sec:pruning_performance}

\begin{table}[tbp]
	\setlength{\tabcolsep}{2.9pt}
	\centering
	\caption{Pruning performance of the approach}
	\shrink
	\begin{tabular}{|c|c||c|c|c|c|}
		\toprule
		\multirow{2}[2]{*}{\begin{sideways}\textbf{\scalebox{.5}{Matching}}\end{sideways}} & \multirow{2}[2]{*}{\textbf{Dataset}} & \textbf{Generated} & \textbf{Trans.} & \textbf{Duplicate} & \textbf{Cache } \\
		&       & \textbf{trans.} & \textbf{to try} & \textbf{trans.} & \textbf{hit ratio} \\
		\midrule
		\midrule
		\multirow{7}[2]{*}{\begin{sideways}\textbf{N-Gram}\end{sideways}} & Web tables &                190,100.8  &            49,560.7  & 52.1\% & 85.4\% \\
		& Spreadsheet &                167,191.6  &            53,924.4  & 45.0\% & 51.0\% \\
		& Open data &             3,628,823.0  &       1,848,653.0  & 49.1\% & 99.0\% \\
		& Synth-50 &                  76,624.0  &            35,552.8  & 52.4\% & 94.8\% \\
		& Synth-50L &                625,475.5  &          148,256.5  & 72.5\% & 96.7\% \\
		& Synth-500 &                584,663.4  &          274,491.2  & 51.8\% & 95.2\% \\
		& Synth-500L &             6,371,427.7  &       1,479,046.5  & 74.1\% & 97.3\% \\
		\midrule
		\midrule
		\multirow{7}[2]{*}{\begin{sideways}\textbf{Golden}\end{sideways}} & Web tables &                  78,922.7  &            30,636.9  & 45.8\% & 74.2\% \\
		& Spreadsheet &                147,049.1  &            50,606.2  & 44.9\% & 51.5\% \\
		& Open data &                794,078.0  &          435,771.0  & 45.1\% & 97.1\% \\
		& Synth-50 &                  90,553.7  &            40,832.4  & 53.1\% & 94.2\% \\
		& Synth-50L &                656,267.0  &          156,242.1  & 72.4\% & 96.3\% \\
		& Synth-500 &                745,167.0  &          344,282.5  & 52.2\% & 95.0\% \\
		& Synth-500L &             6,874,889.8  &       1,602,243.3  & 73.7\% & 96.6\% \\
		\bottomrule
	\end{tabular}%
	\label{tab:cache}%
%	\shrink
%	\shrink
\end{table}%

A challenge in searching for transformations that cover the input is the huge search space. Even though the search space is significantly reduced by introducing placeholders and maximal-length placeholders, our algorithm still has to go through a large number of transformations before compiling a set that can be returned. It is critical to eliminate or reduce early in the process the transformations that cannot contribute to an answer. We employ two pruning strategies: (1) removing duplicate transformations, which arise when transformations are generated independently for each row and the same transformation can be generated for different rows; (2) early detecting transformations that cannot cover a row hence reducing the search space for each row. In this section, we evaluate the effectiveness of our pruning strategies.

Table~\ref{tab:cache} summarizes the number of generated transformations for our datasets and the effectiveness of our pruning. In our real-world data, about half of the transformations are duplicates and are removed before being applied on any input rows. The fraction of duplicates is more than 50\% in most cases and increases as the input length gets larger. For instance, while, in the ``Synth-500L'' table, the length of input is about twice the length  in ``Synth-500'', the number of generated transformations is about 8 times more. However, the number of transformations after removing duplicates is only 4 times larger for `Synth-500L.''. Hence removing duplicates significantly cuts on the number of transformations to be considered.

In terms of early detecting transformations that cannot cover a row, our
cache hit ratio is more than 50\% in all datasets and more than 90\% in many datasets. While the number of transformations can be relatively large, a great portion of them are filtered by our cache-based pruning, which keeps for each row units that cannot cover the row. The savings in runtime is not shown here, but our experiment on the web tables dataset, where the cache hit ratio is 74\%, indicates that the runtime of the approach with our cache-based pruning drops to 61\% of the time it takes for the code to run without. Clearly, the cache-based pruning is effective in reducing the running time.

\subsection{Scalability}
In this section, we evaluate the scalability of our algorithm as the number and length of input rows increase.

\noindent
\textbf{Pruning performance varying the input size.}
In one experiment, we fixed the number of rows at 100 while varying the length of the rows from 20 to 280 characters. As the input length increases, there is a stronger chance that an arbitrary text in target is found in the source. This leads to a larger number of transformations, and many of them are duplicates.
As shown in Figure~\ref{fig:len_cache_dup}, the fraction of duplicate transformations increases up to a point where more than 98\% of the transformations are duplicates, and they are all removed at the generation phase. Removing duplicates deals with the growth in the number of transformations when dataset size grows horizontally.
Our cache hit ratio is more than 90\% and it remains relatively high as we increase the input length.
For example, in one experiment with only 100 input rows and an average row length of 100 characters, a total of 22.6 million transformations were generated among which 20.9 million (or 92.4\%) were detected as duplicates and were removed, leaving us only 1.7 million transformations to evaluate. Applying those remaining transformations to 100 input rows generates 170 million candidates, of which 160 million are removed via our early transformation detection pruning, leaving us with only 1.9 million possible mappings (or a cache hit ratio of 94.4\%) to be evaluated. If no pruning strategy was utilized, one would need to apply about 2.2 billion transformations instead. In the same experiment but with a table that had an average row length of 200 characters, our initial 77.7 million transformations were reduced to 2.5 million after removing duplicates, and the number of trials were reduced from 250 million to only 2.5 million. It can be observed that while the number of generated transformations is increased by a factor of 4 when the input length is doubled, the number of transformations to be applied only increases by a factor of 1.3 due to our pruning strategies.

\begin{figure}[tbp]
	\centering
	\includegraphics[width=0.97\linewidth]{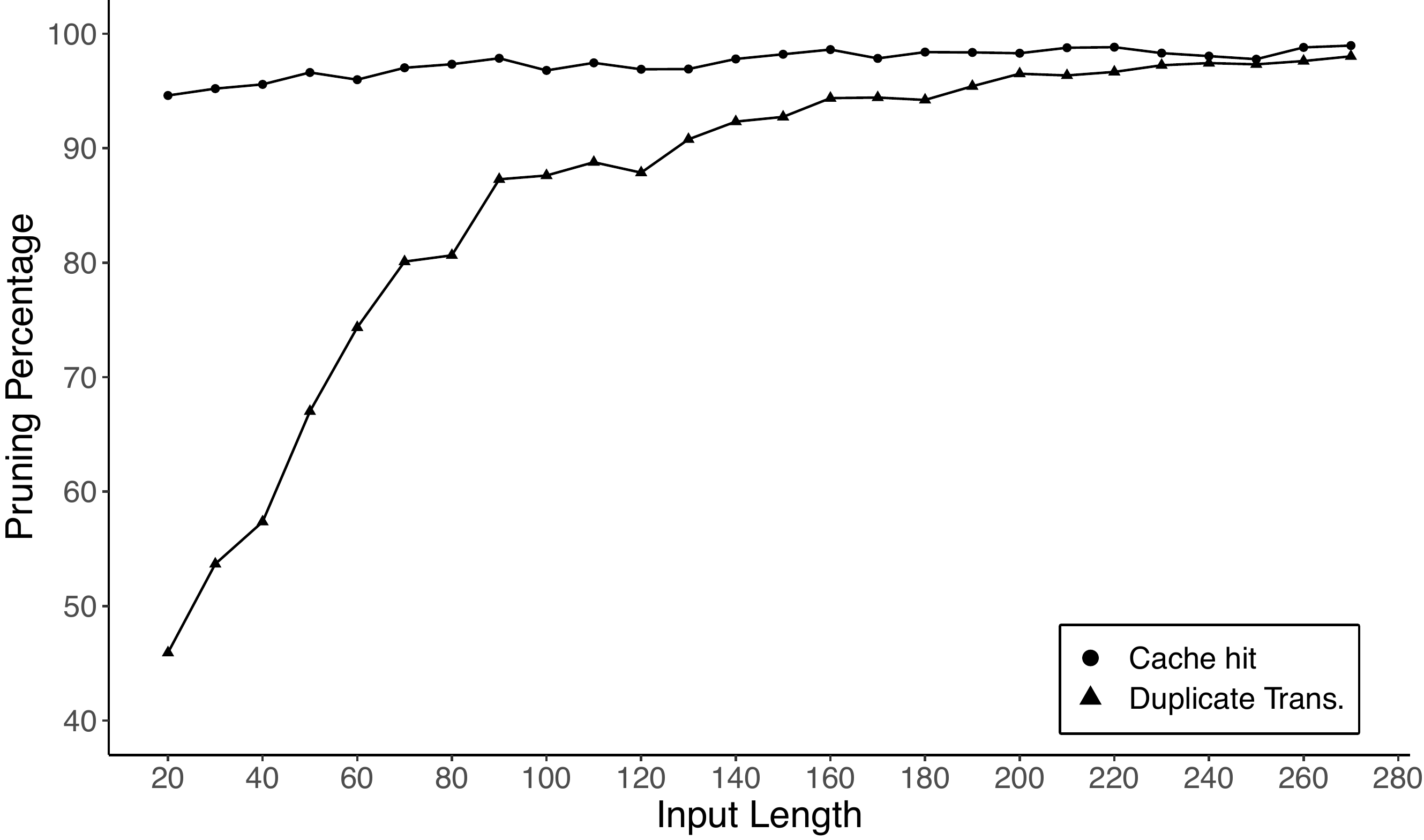}
%	\shrink 
%	\tinyskip
	\caption{Effect of pruning on various input sizes}
	\label{fig:len_cache_dup}
	\shrink
\end{figure}

In another experiment (not reported here), we fixed the input length and varied the number of rows, and the fraction of filtered transformations remained relatively stable. More specifically, 50-60\% of the generated transformations were duplicates that were removed, and the cache hit ratio for non-covering transformations was more than 90\% for all sizes.

\begin{figure*}[]
	\centering
	\begin{subfigure}[b]{0.45\textwidth}
			\centering
			\includegraphics[width=1\linewidth]{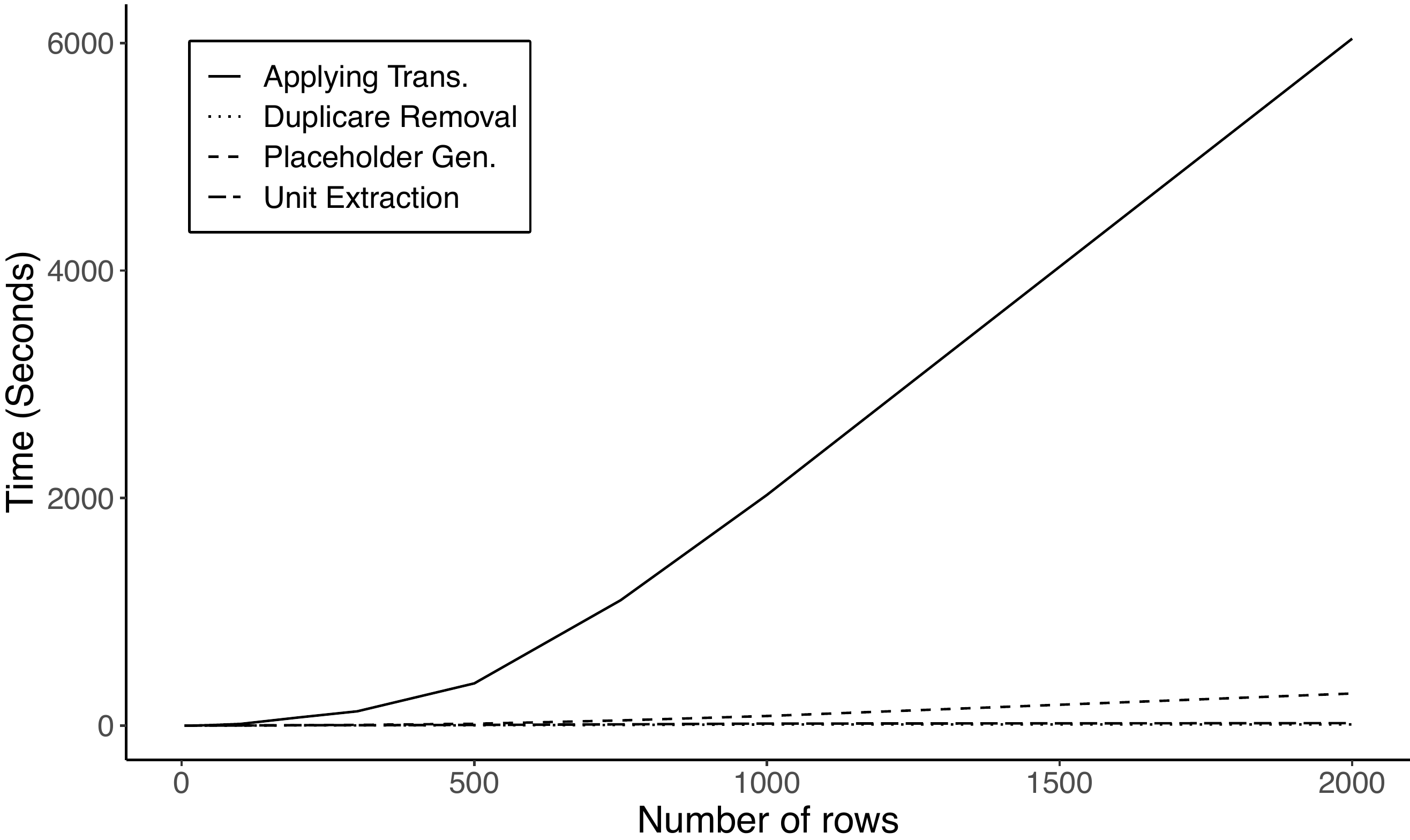}
			\caption{Vertical growth in the dataset}
			\label{subfig:row_times}
		\end{subfigure}
	~ \hspace{20pt}
	\begin{subfigure}[b]{0.45\textwidth}
			\centering
			\includegraphics[width=1\linewidth]{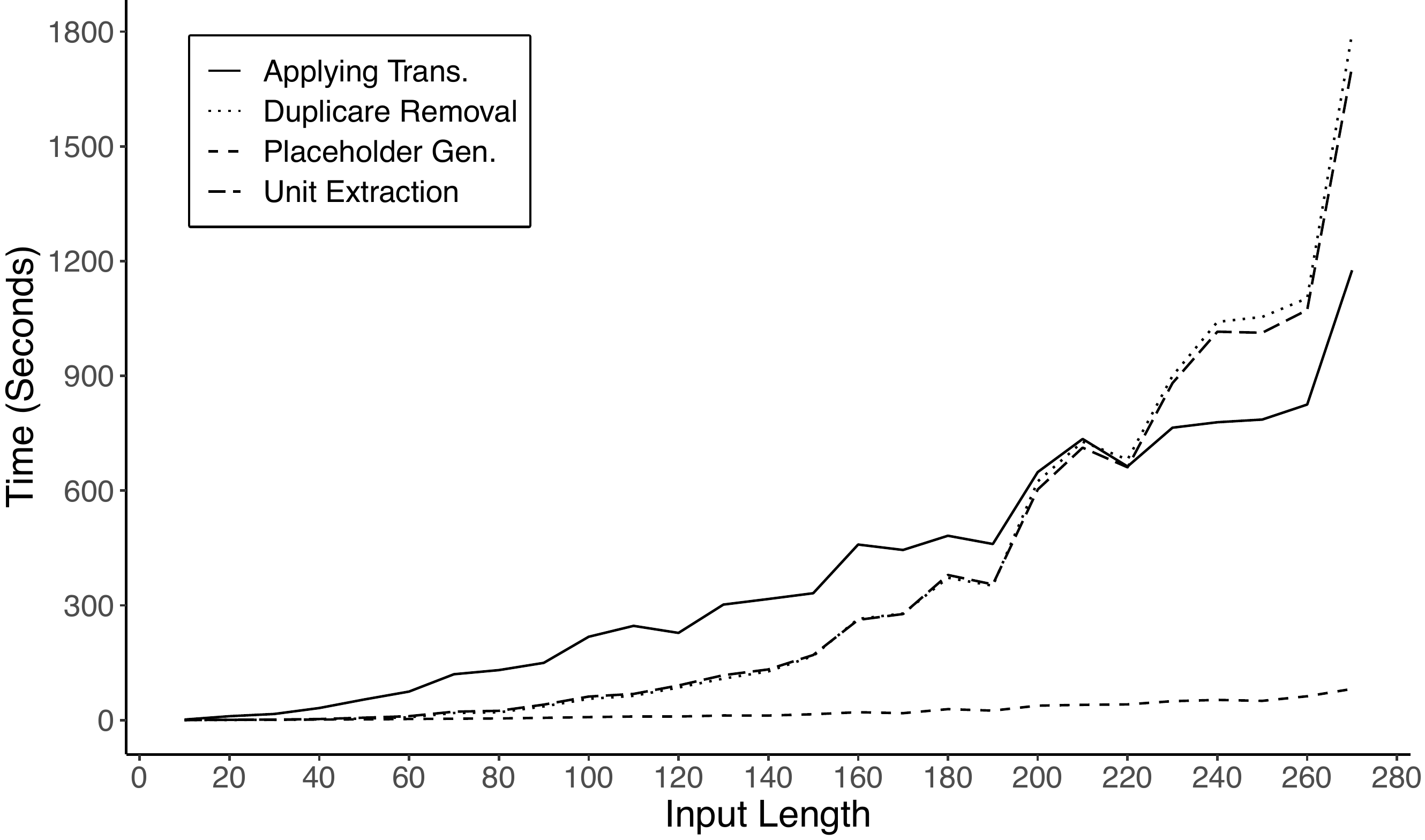}
			\caption{Horizontal growth in the dataset}
			\label{subfig:len_times}
		\end{subfigure}%
	\caption{Runtime for each module on various dataset sizes }
	\label{fig:times}
\end{figure*}

\noindent
\textbf{Running time varying the number of rows.}
Figure~\ref{subfig:row_times} shows the running time of our algorithm, broken down to different modules, with the length of the rows set at 28 and the number of rows varied in our synthetic data. As expected, the largest portion of the time is spent on applying the transformations, and this time is expected to increase somewhat quadratically with the input length. However, the pruning approaches are very effective in significantly reducing this time, and keeping the time curve closer to linear. It should be noted that one does not need to look at a whole dataset to find the transformations, and a small sample is often sufficient to discover all transformations with a relatively high coverage (as discussed in Section~\ref{sec:sampling}).

\noindent
\textbf{Running time varying the length of rows.}
Figure~\ref{subfig:len_times} shows the running time of our algorithm, again broken down to different modules, with the number of rows set at 100 and the input length varied.
With no pruning, the running is expected to be cubic on the input length (with the factor $l^{p=3}$, as discussed in Section~\ref{sec:comlexity}). However, our pruning approaches significantly reduce this time. More interestingly, when the input length passes a certain point (as shown in the figure), the time needed for applying transformations, which is the most time-consuming step of the approach, becomes  even less than the time for placeholder generation and duplicate removal. This is due to the increase in the number of duplicate transformations as the input length grows as well, as the cache hit ratio which is more than 95\% for non-covering transformations. Although duplicate removal is done via a hash set and is expected  to be $O(1)$, it still needs a larger portion of the time because the number of generated transformations increases significantly, and an exact comparison of the transformations is needed when the hashes match.

\section{Other Related Works}
\label{sec:related_works}
The literature related to our work can be grouped into (1) data preparation and cleaning, (2) finding related tables, (3) finding joinable rows, and (4) transforming tables based on examples.

\noindent
\textit{Data cleaning}
Data gathered from the web and spreadsheets may need to be preprocessed and cleaned before querying. Several studies have been conducted on extracting tables from web sources~\cite{WebTables,Automatic2013} and transforming the layout of spreadsheets into relational tables~\cite{foofah,FlashRelate,FlashFill}. Those cleaning steps may be expressed as some transformations. For example, Jin et al.~\cite{foofah} take a set of input-output examples and develop a greedy approach that finds a minimum-length transformation to map the input to the output.
The focus of this line of work is on data cleaning and transforming spreadsheets to relational tables (e.g., mapping multi-row entity descriptions into a single row) and not on joinability. These works are orthogonal to ours and may be utilized as data preprocessing steps in our approach.

\noindent
\textit{Finding related tables}
There is a large body of work on finding related tables that can be joinable~\cite{josie,zhu2016lsh,WebTables}, unionable~\cite{Table_union2018Nargesian,ling2013synthesizing}, or semantically similar~\cite{table2vec,Weblenz,zhang2021semantic}. Some of these works may be used, before applying our approach, to find tables that are joinable under some transformations. For instance, in JOSIE~\cite{josie} the authors define an efficient algorithm that utilizes the set overlap similarity for finding joinable tables. The algorithm may be extended to support the joinability under some transformations. Unlike this line of work,
the focus of our study is on transforming tables for joinability and not on finding relevant tables.

\noindent
\textit{Finding joinable rows}
Finding joinable row pairs when the values do not exactly match is a well-studied subject in the literature~\cite{Chaudhuri2003Robust,Chaudhuri2006fuzzy,fastjoin,MassJoin,yu2016string,auto-em}. The past works utilize various techniques ranging from  string similarity and token matching~\cite{Chaudhuri2003Robust,fastjoin} to entity resolution~\cite{auto-em}. Many of these works perform some forms of fuzzy join, where the focus is on finding the connection points or joinable pairs and not on learning interpretable transformations. For example, having the connection points between a subset of the rows from the tables being joined does not provide the connection points between the rest of the rows in those tables. This line of work is also orthogonal to ours and may be used as a preprocessing step before learning some transformations.

\noindent
\textit{Transforming tables based on examples}
There is some work on generating text transformations to perform a match~\cite{autojoin,autotransform,TDE,FlashFill,BlinkFill}. Some of these approaches exploit Programming By Example (PBE) techniques~\cite{FlashFill,BlinkFill,foofah} where the user provides a set of examples, and transformations are learned from those examples. One issue with these PBE methods is that the examples are not always available or may not be all correct, and it can be time-consuming for domain experts to provide them manually. Also, a small set of examples may not cover all transformations that are required for mapping source and target columns.
For example, FlashFill~\cite{FlashFill}, a pioneer in PBE based approaches, is designed to find a single transformation covering all examples and will not yield any transformations when datasets tend to have noise and may inherently need more than a single transformation to be covered, which is the case with most of real-world datasets.
FlashFill works fine mainly when the user manually provides a small reliable set of examples that are correct, which is not the case in our study. Unlike the PBE based methods, in our approach, a larger number of examples are automatically detected. Because of this auto-detection aspect of examples and the intricacies of integrating data from different sources, the examples can have noise and may not be covered using one transformation. Furthermore, building and intersecting directed graphs that FlashFill uses are computationally expensive when the number of examples grows.
On the other hand, our approach, with a focus on scalability, is considerably more robust to noise and inaccuracy in examples.
Another group of studies build search engines that collect the transformations from the web, GitHub repositories, and other sources~\cite{TDE,autotransform,DataXFormer}. This line of work is orthogonal to ours, in terms of the domains that it can be applied, and may require a significant amount of resources and may not provide much flexibility in terms of choosing, limiting, and customizing the transformations. Finally, our work falls within the class of studies that generate transformations automatically without requiring users to provide examples. Auto-join~\cite{autojoin} is one of the state-of-the-art approaches in this area, and we extensively compare our work to Auto-join. A key difference between auto-join and our approach is that auto-join is a back-tracking method that blindly searches the transformation space, while our approach exploits the textual evidences to shrink the search space.

\section{Conclusion, Limitations and Future Work}
\label{sec:conclusions}
We have studied the problem of efficiently joining textual data under the condition that the join columns are not formatted the same and cannot be equi-joined, but they become joinable under some transformations. We have developed an efficient algorithm over a rich set of basic operations that can be composed to form transformations. We have conducted both analytical and empirical evaluation of our algorithm and have compared its performance to a state-of-the-art approach. Our evaluation reveals that not only our algorithm covers a rich set of transformations but is also a few orders of magnitude faster than our competitor. 

Based on our analysis and experiments on benchmark datasets, the main limitations of our work can be broken down into the following three cases: 
(1) There is a gap between the performance of a golden row matching and our row matching algorithm (as shown in Table~\ref{tab:row_matching}), and this gap varies for different datasets. This gap translates to noise that is passed to the transformation discovery, which can produce false transformations. 
(2) As mentioned in Section~\ref{sec:exp_setup}, we limit the number of placeholders, which improves the running performance but some transformations can be missed.
(3) Some tables cannot be joined using string-based transformations only. For example, our transformations cannot capture semantic relations such as when one value is a synonym of another value.

Our approach can be extended in a few directions. One direction is transfer learning where transformations obtained for one dataset can be adapted to another dataset. 
%Extending our work to employ external resources such as a knowledge base is another interesting direction.	
Extending our work to employ non-string-based transformations considering semantics of the words is another interesting direction that deals with the limitations of textual transformations.

	\bibliographystyle{ACM-Reference-Format}
	\bibliography{ref}
	
	%\appendix
	
\end{document}